\documentclass[preprint,12pt]{elsarticle}

\usepackage{graphics}

\usepackage{amssymb}
\usepackage{amsthm}

\usepackage{epstopdf}
\usepackage{algorithm}
\usepackage{algorithmic}
\usepackage{fullpage}

\newtheorem{theorem}{Theorem}
\newtheorem{lemma}{Lemma}

\journal{arXiv}

\begin{document}

\begin{frontmatter}



\title{A Note on the Inapproximability of Induced Disjoint Paths}

\author{Gaoxiu Dong\corref{G.Dong}}
\ead{dgx\_01@outlook.com}
\author{Weidong Chen\corref{W.Chen}}
\cortext[W.Chen]{Corresponding Author. Tel.: +86 13005158508; Fax: +86 20 85215418.}
\ead{chenwd@scnu.edu.cn}
\address{Department of Computer Science, South China Normal University, Guangzhou 510631, China}

\begin{abstract}

We study the inapproximability of the induced disjoint paths problem on an arbitrary $n$-node $m$-edge undirected graph, which is to connect the maximum number of the $k$ source-sink pairs given in the graph via induced disjoint paths. It is known that the problem is NP-hard to approximate within $m^{{1\over 2}-\varepsilon}$ for a general $k$ and any $\varepsilon>0$. In this paper, we prove that the problem is NP-hard to approximate within $n^{1-\varepsilon}$ for a general $k$ and any $\varepsilon>0$ by giving a simple reduction from the independent set problem.
\end{abstract}


\end{frontmatter}



\section{Problem and Hardness of Approximation}\label{sec-1}\

The concept of induced disjoint paths is recently introduced to characterize some non-interfering situation during data transmission in wireless networks \cite{K. Kawarabayashi 2008, K. Zhang 2011}. A set of paths in an undirected graph (graph, for short) are called induced disjoint paths if each one of them has no chords (i.e., is an induced path) and any two of them have neither common nodes nor adjacent nodes. Given a graph $G$ and a collection of $k$ source-sink pairs in $G$ ($k$ as part of the input of the problem), the induced disjoint paths problem (IDPP) is to connect the maximum number of these source-sink pairs via induced disjoint paths. Zhang et al. \cite{K. Zhang 2011} have shown that for any $\varepsilon>0$, it is NP-hard to approximate IDPP to within $m^{{1\over 2}-\varepsilon}$ on an arbitrary $n$-node $m$-edge graph, and there is a greedy algorithm with approximation ratio $\sqrt{m}$, matching the lower bound $\sqrt{m}$. In this paper, we will prove that for any $\varepsilon >0$, it is NP-hard to approximate IDPP to within $n^{1-\varepsilon}$ on an arbitrary $n$-node graph. Our method is based on a simple approximation-preserving reduction from the independent set problem to IDPP, and an inapproximability result for finding maximum independent sets in general graphs.

IDPP is an extension of the node-disjoint paths problem (DPP), one simple version of the disjoint paths problem that is one of the classic NP-hard combinatorial optimization problems. In DPP, we are given a graph $G$ and a set of $k$ source-sink pairs in $G$ ($k$ as part of the input of the problem), and the objective is to find a largest subset of the source-sink pairs that can be simultaneously connected in an node-disjoint manner. In fact, any instance of DPP can be reduced to an instance of IDPP by subdividing every edge into two edges. Thus, IDPP is harder than DPP. As a comparison, the best known approximation guarantee for DPP on an arbitrary $n$-node graph is $O(\sqrt{n})$  \cite{C. Chekuri 2006, T. Nguyen 2007}. Recently Andrews et al. \cite{M. Andrews 2010} have shown that for any $\varepsilon >0$, DPP is hard to approximate within $\log^{{1\over 2}-\varepsilon}n$ on an arbitrary $n$-node graph, unless NP$\subseteq$ZPTIME($n^{poly(\log n)}$). This might be the best known inapproximability result for DPP on general graphs. An interesting open question is whether there exists an $\varepsilon>0$ so that there is no polynomial time $O(n^\varepsilon)$-approximation algorithm for DPP on an arbitrary $n$-node graph, unless P=NP \cite{J.M. Kleinberg 1996}.

\section{New Result on Hardness of Approximation}\label{sec-2}\
This section is devoted to the proof of a new approximation hardness result for IDPP on general graphs.
The main theorem is as follows.

\begin{theorem}\label{theorem-1}
On an arbitrary $n$-node graph, there can be no polynomial time algorithm that approximates IDPP to within $n^{1-\varepsilon}$ for any $\varepsilon>0$, unless P=NP.
\end{theorem}

For the purpose of clarification, we prove below two lemmas that together immediately lead to Theorem \ref{theorem-1}.
\begin{lemma}\label{lemma-1}
On an arbitrary $n$-node graph, if there exists a polynomial time algorithm for IDPP with approximation ratio $n^{1-\varepsilon}$ for a certain $\varepsilon>0$, then there exists a polynomial time algorithm for IDPP with approximation ratio $({n\over3})^{1-\varepsilon'}$ for some $\varepsilon'>0$.
\end{lemma}
\begin{proof}
Assume that we have a polynomial time algorithm app-IDPP with approximation ratio $n^{1-\varepsilon}$ for a certain $\varepsilon>0$. We now describe a polynomial time algorithm and claim that the algorithm has approximation ratio $({n\over3})^{1-\varepsilon'}$ for some $\varepsilon'>0$. The algorithm uses as input an instance of IDPP, and calls a procedure to solve the instance according to whether or not $n$ is smaller than $3^{1+{1\over\varepsilon}}$.

\textbf{Case 1} ($n<3^{1+{1\over\varepsilon}}$): It calls a brute-force procedure to solve the instance. This can be done in a constant time that depends on $1\over \varepsilon$.

\textbf{Case 2} ($n\ge 3^{1+{1\over\varepsilon}}$): It calls app-IDPP to solve the instance. This can be done in polynomial time according to the above assumption.

It is evident that the algorithm has a polynomial time complexity, and the solution yielded in Case 1 is optimal. Therefore, the approximation ratio of the algorithm relies on the solution yielded by app-IDPP in Case 2. Since app-IDPP is an $n^{1-\varepsilon}$-approximation algorithm, so is the algorithm. In order to prove that the algorithm has the claimed approximation ratio, $({n\over3})^{1-\varepsilon'}$, it suffices to show that there exists $\varepsilon'>0$ such that $({n\over3})^{1-\varepsilon'}\ge n^{1-\varepsilon}$. For this purpose, letting $\varepsilon'=\varepsilon^2$, we readily have from $n\ge 3^{1+{1\over\varepsilon}}$
\begin{displaymath}
n^{(1-\varepsilon)\varepsilon}\ge 3^{(1-\varepsilon)(1+\varepsilon)},
\end{displaymath}
\begin{displaymath}
n^{1-\varepsilon^2}\ge 3^{1-\varepsilon^2}n^{1-\varepsilon},
\end{displaymath}
\begin{displaymath}
{({n\over 3})}^{1-\varepsilon'}\ge n^{1-\varepsilon}.
\end{displaymath}
This completes the proof.
\end{proof}

The complexity of approximating IDPP can be related to that of finding maximum independent sets in graphs by a reduction from the independent set problem to IDPP in the proof of Lemma \ref{lemma-2} below. Recall that an independent set of a graph is a set of nodes no two of which are connected by an edge, and the independent set problem is to find a maximum independent set in the graph. The independent set problem and
the clique problem are complementary: a clique in $G$ is an independent set in the complement graph of $G$ and vice versa. It is known that on an
arbitrary $n$-node graph, it is NP-hard to approximate the maximum clique problem to within $n^{1-\varepsilon}$ for any $\varepsilon>0$ \cite{D. Zuckerman
 2007}. This inapproximability result for the clique problem is applied equally well to the independent set problem. Now we prove the following lemma.

\begin{lemma}\label{lemma-2}
On an arbitrary $n$-node graph, there can be no polynomial time algorithm that approximates IDPP to within $({n\over3})^{1-\varepsilon}$ for any $\varepsilon>0$, unless P=NP.
\end{lemma}
\begin{proof}
We consider a reduction from the independent set problem to IDPP. Given an instance of the independent set problem, a graph $G'$ with $n'$ nodes and $m'$ edges, we below construct an instance of IDPP, which contains a graph $G$ and $k$ pairs of nodes in $G$ .

$\bullet$ From the graph $G'$, we form the graph $G$ by adding $n'$ pairs of nodes to $G'$, one pair for each node in $G'$, and adding $2n'$ edges to $G'$ to join each pair of added nodes to the corresponding node in $G'$. The resulting graph $G$ has $3n'$ nodes and $m'+2n'$ edges.

$\bullet$ Letting $k=n'$, we form $k$ pairs of nodes in $G$ by selecting the $n'$ pairs of added nodes.

The above reduction can be done in polynomial time. Moreover, it is obvious that the given instance has an independent set of size $t$ if and only if the constructed instance of IDPP has $t$ induced disjoint paths. This implies that the reduction is an approximability-preserving one. By the above inapproximability result for the independent set problem on a general graph, unless P=NP, there does not exist a polynomial time algorithm for approximating IDPP to within a factor of $({n\over3})^{1-\varepsilon}$ of the optimal solution since $G$ has $3n'$ nodes. The proof is complete.
\end{proof}

\section{Conclusion}\label{sec-7}\
A new result on hardness of approximation for IDPP on general graphs has been proved in this paper. That is, on an arbitrary $n$-node $m$-edge graph,
for all $\varepsilon>0$, approximating IDPP to within $n^{1-\varepsilon}$ is NP-hard. The new result is distinct from the known one in the literature that states that for all $\varepsilon>0$, approximating IDPP to within $m^{{1\over 2}-\varepsilon}$ is NP-hard. Since $m^{1/2}<n$ always holds for any $n$-node $m$-edge graph, the new result yields a stronger lower bound on the best possible approximation guarantee for IDPP than the known result can yield, namely, $m^{{1\over 2}-\varepsilon}<n^{1-\varepsilon}$ for any $\varepsilon>0$. On those graphs with $m=O(n^{\alpha})$ for $\alpha<2$ (for example, sparse graphs), the new result obviously can give such a tighter lower bound.

\section*{Acknowledgements}
This research is supported by the National Natural Science Foundation of China (No. 61370003), and by the Scientific Research Foundation for the Returned Overseas Chinese Scholars, State Education Ministry.





\bibliographystyle{model1-num-names}

\end{document}